\documentclass[12pt]{amsart}
\usepackage{graphicx,amscd, color,amsmath,amsfonts,amssymb,geometry,amssymb}
\usepackage[initials]{amsrefs}

\newtheorem{theorem}{Theorem}[section]
\newtheorem{proposition}[theorem]{Proposition}
\newtheorem{corollary}[theorem]{Corollary}

\newtheorem{remark}[theorem]{Remark}

\newtheorem{lemma}[theorem]{Lemma}

\newtheorem{definition}[theorem]{Definition}
\numberwithin{equation}{section}

\begin{document}
\title{Does symmetry Imply PPT Property?}
\author[Cariello ]{D. Cariello}
\thanks{PACS numbers: 03.65.Ud, 03.67.Mn}
\thanks{Key words and Phrases: Symmetric State, Hermitian Schmidt Decomposition, PPT, SPC, Tensor Rank}

\address{Faculdade de Matem\'atica, \newline\indent Universidade Federal de Uberl\^{a}ndia, \newline\indent 38.400-902 Ð Uberl\^{a}ndia, Brazil.}\email{dcariello@famat.ufu.br}

\address{Departamento de An\'{a}lisis Matem\'{a}tico,\newline\indent Facultad de Ciencias Matem\'{a}ticas, \newline\indent Plaza de Ciencias 3, \newline\indent Universidad Complutense de Madrid,\newline\indent Madrid, 28040, Spain.}
\email{dcariell@ucm.es}

\keywords{}

\subjclass[2010]{}

\maketitle

\begin{abstract} Recently, in \cite{cariello}, the author proved that many results that are true for PPT matrices also hold for another class of matrices with a certain symmetry in their Hermitian Schmidt decompositions. These matrices were called SPC in \cite{cariello} $($definition \ref{defSPC}$)$. Before that, in \cite{guhne}, T\'oth and G\"uhne proved that if a state is symmetric then it is PPT if and only if it is SPC.  A natural question appeared: What is the connection between SPC matrices and PPT matrices?
Is every SPC matrix PPT?

Here we show that every SPC matrix is PPT in $M_2\otimes M_2($theorem \ref{SPC}$)$. This theorem is a consequence of the fact that every density matrix in  $M_2\otimes M_m$, with tensor rank smaller or equal to 3, is separable $($theorem \ref{tensorrank3}$)$. This theorem is a generalization of the same result found in \cite{cariello} for tensor rank 2 matrices in $M_k\otimes M_m$.

Although, in $M_3\otimes M_3$, there exists a SPC matrix with tensor rank 3 that is not PPT $($proposition \ref{example}$)$. We shall also provide a non trivial example of a family of matrices in $M_k\otimes M_k$, in which both, the SPC and PPT properties, are equivalent $($proposition \ref{nontrivial}$)$. Within this family, there exists a non trivial subfamily in which the SPC property is equivalent to separability $($proposition \ref{exampleseparable}$)$.

\end{abstract}

\section*{Introduction}

The PPT property is an important concept in Quantum Information Theory. Since the PPT property was noticed to be a   necessary condition for separability of density matrices $($\cite{Peres}$)$, many papers were published regarding applications or characterizations of PPT property, e.g., \cite{kossakowski}, \cite{kossakowski2}, \cite{hildebrand}, \cite{horodecki},  \cite{Steinhoff}, \cite{guhne}. 

The most important feature of this property  was proved by Horodecki in \cite{horodecki}: The PPT property is equivalent to separability in the space $M_2\otimes M_m$, $m=2,3$.

We can refer, for example, to the following papers devoted to find classes of PPT matrices: \cite{kossakowski}, \cite{kossakowski2}, \cite{Steinhoff}.

With respect to papers devoted to characterize the PPT property by means of other properties, one example is Hildebrand's work \cite{hildebrand}. He found a necessary and sufficient condition for an operator acting on a $n-$dimensional Hilbert space, $H_n$, to be PPT in any possible decomposition of $H_n$ as $H_k\otimes H_m$, for $n=km$. The analogous result for the case $n=4$ was proved in \cite{Verstraete}.

Another example is the paper \cite{guhne} of T\'oth and G\"uhne. They defined a symmetric state $\rho$, as a state that satisfies $\rho T=T\rho=\rho$, where $T$ is the flip operator. They showed that $\rho$ is PPT if and only if the Hermitian Schmidt decomposition of $\rho$ is $\sum_{i=1}^n\lambda_i\gamma_i\otimes\gamma_i$, with $\lambda_i>0$ .

In \cite{cariello}, the author noticed that even if we remove the hypothesis of $\rho$ being symmetric, in the sense of T\'oth and G\"uhne, the positive matrices with that symmetric Hermitian Schmidt decomposition share many properties with PPT matrices. Matrices with that symmetric Hermitian Schmidt decomposition were called  SPC matrices in \cite{cariello} $($See definition \ref{defSPC}$)$. 

The author of \cite{cariello} proved that the following results hold for SPC and PPT matrices:

\begin{enumerate}
\item If a SPC matrix $A$ has the Hermitian Schmidt decomposition $\sum_{i=1}^n\gamma_i\otimes\gamma_i$ $( \lambda_i=1$ for every i$)$ then A is separable. \\
If a PPT matrix $B$ has the Hermitian Schmidt decomposition $\sum_{i=1}^n\gamma_i\otimes\delta_i$ $($all the coefficients equal to 1$)$ then $B$ is separable.
\item SPC/PPT matrices have Split Decompositions.
\item SPC/PPT matrices are weakly irreducible or a sum of weakly irreducible SPC/PPT matrices.
\item The descriptions of weakly irreducible  SPC/PPT matrices are similar.
\item There are sharp inequalities providing separability for SPC/PPT matrices. 
\end{enumerate}

 After all this evidence, we shall make a question: Is every SPC matrix PPT?
  This paper is devoted to the study of this question. 
 
 We show in section 4 that every SPC matrix in $M_2\otimes M_2$ is PPT $($theorem \ref{SPC}$)$ and separable by Horodecki's theorem. Thus, in some sense, symmetry implies separability in $M_2\otimes M_2$. In order to obtain this result, we prove in section 3 that every positive semidefinite matrix in $M_2\otimes M_m$ with tensor rank smaller or equal to 3 is separable$($theorem \ref{tensorrank3}$)$.  It was proved in \cite{cariello} that every positive semidefinite matrix in $M_k\otimes M_m$ with tensor rank 2 is separable. Thus, our result regarding tensor rank 3 matrices, generalizes this result for the space $M_2\otimes M_m$. We prove that both results can not be extended to higher dimension. As a matter of fact, we show in section 5 that exists a SPC matrix in $M_3\otimes M_3$ with tensor rank 3 which is not PPT, therefore it is not separable $($proposition \ref{example}$)$.
 We obtain  these results using properties of the linear tranformation $S$ defined in \ref{defS}. 
 We give a very simple proof of T\'oth and G\"uhne's theorem using properties of this $S$ and, finally, we show a non trivial example of a family of matrices in $M_k\otimes M_k$, in which the SPC property and  the PPT property are equivalent $($proposition \ref{nontrivial}$)$ and inside this family, we discover a non trivial subfamily in which the SPC property and separability are equivalent $($proposition \ref{exampleseparable}$)$. 
 
\section{Preliminary Results and Definitions}

In this section we provide the definitions and the preliminary results  used in the main results of this paper. Lemmas \ref{prop1S} and \ref{hermitianeigenvector} are used quite a few times.

Let $M_k$ denote the set of complex matrices of order $k$. We shall identify the tensor product space $\mathbb{C}^n\otimes\mathbb{C}^k$ with $\mathbb{C}^{nk}$ and the tensor product space $M_{k}\otimes M_{m}$ with $M_{km}$, via Kronecker product. It allow us to write $(v\otimes w)(r\otimes s)^t= vr^t\otimes ws^t$, where $v\otimes w$ is a column and $(v\otimes w)^t$ its transpose.  Therefore if $x,y\in\mathbb{C}^n\otimes\mathbb{C}^m$ we have $xy^t\in M_n\otimes M_m$.
The trace of a matrix $A$ is denoted by $tr(A)$ and $A^t$ shall stand for the transpose of $A$.

\begin{definition}\label{defSPC}$($\textbf{SPC matrices}$)$ Let $A\in M_k\otimes M_k\simeq M_{k^2}$ be a positive semidefinite Hermitian matrix. We say that $A$ is SPC, if $A$ has the following symmetric Hermitian Schmidt decomposition with positive coefficients: $\sum_{i=1}^n\lambda_i\gamma_i\otimes\gamma_i$, with $\lambda_i>0$, for every $i$.
\end{definition}
\begin{remark} The SPC matrices can be defined using only the concept of Hermitian decomposition. See corollary \ref{defSPC2} for a simpler description.
\end{remark}

\begin{definition}\label{defPPT}$($\textbf{PPT matrices}$)$ Let $A=\sum_{i=1}^nA_i\otimes B_i\in M_k\otimes M_m\simeq M_{km}$ be a positive semidefinite Hermitian matrix. We say that $A$ is positive under partial transposition or simply  PPT, if $A^{t_2}=Id\otimes(\cdot)^t(A)=\sum_{i=1}^nA_i\otimes B_i^t$ is positive semidefinite.
\end{definition}

\begin{definition}\label{definitionseparability}\textbf{$($Separable Matrices$)$} Let $A\in M_k\otimes M_m$.
We say that $A$ is separable if $A=\sum_{i=1}^n C_i\otimes D_i$ such that   $C_i\in M_k$ and $D_i\in M_m$ are positive semi-definite Hermitian matrices for every $i$.
\end{definition}

\begin{definition}\label{definition1} 
\begin{enumerate}
\item Denote by $A\circ B$ the Schur product of $A,B\in M_k$.
\item Let $T\in M_k\otimes M_k$ be the flip operator, i.e.,\\ $T(a\otimes b)=b\otimes a$, for every $a,b\in\mathbb{C}^k$.
\item Let $F:M_k\rightarrow \mathbb{C}^k\otimes\mathbb{C}^k$,  $F(\sum_{i=1}^n a_ib_i^t)=\sum_{i=1}^na_i\otimes b_i$. 
\item We say that $v\in\mathbb{C}^k\otimes\mathbb{C}^k$ is Hermitian if $F^{-1}(v)\in M_k$ is Hermitian.

\end{enumerate}
\end{definition}

\begin{definition}\label{defS} Let $S:M_{k}\otimes M_k\rightarrow M_{k}\otimes M_k$ be defined by 
$$S(\sum_{i=1}^nA_i\otimes B_i)= \sum_{i=1}^nF(A_i)F(B_i)^t.$$
\end{definition}

\begin{lemma}\label{prop1S}  Let $S:M_{k}\otimes M_k\rightarrow M_{k}\otimes M_k$ be the linear transformation defined in \ref{defS}. Let $v_i,w_i\in\mathbb{C}^k\otimes\mathbb{C}^k$ then $S(\sum_{i=1}^nv_iw_i^t)=\sum_{i=1}^n F^{-1}(v_i)\otimes F^{-1}(w_i)$ and $S^2=Id:M_{k}\otimes M_k\rightarrow M_{k}\otimes M_k$.
\end{lemma}

\begin{proof} Since $S$ is  a linear tranformation, we just need to prove the formula for $n=1$. Since $S(vw^t)$ and $F^{-1}(v)\otimes F^{-1}(w)$  are linear on the variables $v$ an $w$,  we just need to show the theorem for $v=a\otimes b$ and $w=c\otimes d$.

Notice that $vw^t=ac^t\otimes bd^t$ and $$S(ac^t\otimes bd^t)=F(ac^t)\otimes F(bd^t)=(a\otimes c)(b\otimes d)^t=ab^t\otimes cd^t.$$

Now $F^{-1}(v)=ab^t$ and $F^{-1}(w)=cd^t$.
Finally, notice that $S^2(\sum_{i=1}^nv_iw_i^t)=S(\sum_{i=1}^nF^{-1}(v_i)\otimes F^{-1}(w_i))=\sum_{i=1}^nv_iw_i^t$.  
\end{proof}
 \begin{remark} \label{spectral} Remind that $F$ is an isometry, i.e., $tr(F(A)\overline{F(B)}^t)=tr(AB^*)$, for every $A,B\in M_k$, and $tr(F^{-1}(v)F^{-1}(w)^*)=tr(v\overline{w}^t)$, for every $v,w\in\mathbb{C}^k$. Therefore,  $A=\sum_{i=1}^n\lambda_i\gamma_i\otimes\overline{\gamma_i}$, such that $\{\gamma_1,\ldots,\gamma_n\}$ is a orthonormal set of matrices and $\lambda_i\in\mathbb{R}$,  if and only if, $\sum_{i=1}^n\lambda_iv_i\overline{v}_i^t$ is a spectral decomposition of $S(A)$, where $F(\gamma_i)=v_i$.
 \end{remark}

\begin{lemma}\label{hermitianeigenvector} Let $A\in M_k\otimes M_k$ be a Hermitian matrix. The following conditions are equivalent:
\begin{enumerate}
\item  $A=\sum_{i}\lambda_i\gamma_i\otimes\gamma_i^t$, such that $\lambda_i$ are real numbers and $\gamma_i$ Hermitian matrices.
\item $A=\sum_{j} \alpha_jv_j\overline{v_j}^t$, such that $\alpha_j$ are real numbers and $v_j$ Hermitian vectors.
\item  Exists a basis of $\mathbb{C}^k\otimes\mathbb{C}^k$ formed by Hermitian eigenvectors of $A$.
\end{enumerate}
\end{lemma}
\begin{proof} $(1)\Rightarrow (3)$
Since $\gamma_i^t=\overline{\gamma_i}$, because $\gamma_i$ is Hermitian, then $Av$ is Hermitian for every Hermitian $v\in\mathbb{C}^k\otimes\mathbb{C}^k$.

Let $w\in\mathbb{C}^k\otimes\mathbb{C}^k$ be an eigenvector of $A$ associated to the eigenvalue $\lambda$. Let $w=w_1+iw_2$, where $w_1,w_2$ are Hermitian vectors. Since $A$ is a Hermitian matrix, $\lambda$ is a real number. Notice that $Aw=Aw_1+iAw_2=\lambda w_1+i \lambda w_2$. 

Now $Aw_1-\lambda w_1=i (\lambda w_2-Aw_2)$. Since $Aw_1-\lambda w_1$ and $\lambda w_2-Aw_2$ are Hermitian vectors, we obtain $0=Aw_1-\lambda w_1=\lambda w_2-Aw_2$.

Thus, every eigenvector of $A$ is a linear combination of Hermitian eigenvectors of $A$. Thus there is a set of Hermitian eigenvectors of $A$ that span a basis for $\mathbb{C}^k\otimes\mathbb{C}^k$ and we may extract a basis from this set.

$(3)\Rightarrow (2)$ Since there is a basis of $\mathbb{C}^k\otimes\mathbb{C}^k$ formed by Hermitian eigenvectors of $A$, we can obtain an orthonormal basis of Hermitian eigenvectors. Therefore we obtain a spectral decomposition $A=\sum_{j} \alpha_jv_j\overline{v_j}^t$, where $\alpha_j$ are real numbers and $v_j$ Hermitian vectors.

$(2)\Rightarrow (1)$ By hypothesis, $A=\sum_{j=1}^n\alpha_jv_j\overline{v_j}^t$, where $\alpha_j$ is a real number and $v_j$ is Hermitian for every $j$, i.e., $F^{-1}(v_j)$ is a Hermitian matrix.
Notice that $F^{-1}(\overline{v_j})=\overline{F^{-1}(v_j)}=F^{-1}(v_j)^t$.

By lemma \ref{prop1S}, we have $S(A)=S(\sum_{j=1}^n\alpha_jv_j\overline{v_j}^t)=\sum_{j=1}^n\alpha_j F^{-1}(v_j)\otimes F^{-1}(v_j)^t$. Notice that $S(A)$ is a Hermitian matrix, since $\alpha_j\in\mathbb{R}$ and  $F^{-1}(v_j)$ is Hermitian for every $j$.

Since we have already proved that $(1\Rightarrow 2)$ then  $S(A)=\sum_{i=1}^m\lambda_i w_i\overline{w_i}^t$, where $w_i$ is Hermitian for every $i$.
By lemma \ref{prop1S}, $A=S^2(A)= S(\sum_{i=1}^m\lambda_i w_i\overline{w_i}^t)=\sum_{i=1}^m\lambda_i F^{-1}(w_i)\otimes F^{-1}(\overline{w_i})$.

Finally, since $F^{-1}(\overline{w_i})=\overline{F^{-1}(w_i)}$ and $F^{-1}(w_i)$ is Hermitian then $\overline{F^{-1}(w_i)}=F^{-1}(w_i)^t$ and $A=\sum_{i=1}^m\lambda_i F^{-1}(w_i)\otimes F^{-1}(w_i)^t$.
\end{proof}

\begin{corollary} \label{defSPC2} Let $A\in M_k\otimes M_k$ be a positive semidefinite Hermitian matrix. $A$ is SPC if and only if $A$ has a Hermitian decomposition $\sum_{i=1}^n \alpha_i A_i\otimes A_i$ with $\alpha_i>0$ for every $i$.
\end{corollary}
\begin{proof}
If $A$ is SPC then it is obvious that $A$ has a Hermitian decomposition required in this corollary. Suppose $A=\sum_{i=1}^n \alpha_i A_i\otimes A_i$, where $A_i$ is Hermitian for every $i$ and $\alpha_i>0$.

Notice that $S(A^{t_2})=S(\sum_{i=1}^n \alpha_i A_i\otimes A_i^t)= S(\sum_{i=1}^n \alpha_i A_i\otimes \overline{A_i})=\sum_{i=1}^n \alpha_i v_i\overline{v_i}^t$ is a positive semidefinite Hermitian matrix with $v_i=F(A_i)$.

By lemma \ref{hermitianeigenvector}, $S(A^{t_2})$ has a spectral decomposition $\sum_{i=1}^m\lambda_iw_i\overline{w_i}^t$ with Hermitian eigenvectors $w_i$. Remind that $\lambda_i>0$.

Therefore, $A^{t_2}=S(\sum_{i=1}^m\lambda_iw_i\overline{w_i}^t)=\sum_{i=1}^m\lambda_i F^{-1}(w_i)\otimes F^{-1}(w_i)^t$. Remind by remark \ref{spectral} that this is a Hermitian Schmidt decomposition of $A^{t_2}$ and thus $\sum_{i=1}^m\lambda_i F^{-1}(w_i)\otimes F^{-1}(w_i)$ is a Hermitian Schmidt decomposition of $A$. Therefore $A$ is SPC.
\end{proof}

\section{T\'oth and G\"uhne's Theorem}

In this paper, we shall employ lemma \ref{hermitianeigenvector} quite several times. This lemma was obtained using some properties $($lemma \ref{prop1S}$)$ of the linear transformation $S$ defined in \ref{defS}.

In this small section, we show that there are other properties of $S\ ($lemma \ref{prop2S}$)$ that can be used, for example, to obtain a very simple proof of theorem \ref{theoremGuhne} obtained by T\'oth and G\"uhne  in \cite{guhne}.
We only need the following two formulas.

\begin{lemma} \label{prop2S} Let $A\in M_k\otimes M_k$ then $A^{t_2}=S(AT)T$ and $S(A^{t_2})=S(A)T$.
\end{lemma}

\begin{proof}
Since the multiplication by $T$, the partial transposition and $S$ are linear transformations acting on $M_{k}\otimes M_k$, we only need to prove these formulas for a set of generators of $M_{k}\otimes M_k$. Thus, let us prove this formula only for matrices of the type $A=ab^t\otimes cd^t$, where $a,b,c,d\in\mathbb{C}^k$. 

Notice that 
$A=ab^t\otimes cd^t=(a\otimes c)(b\otimes d)^t$, $AT=(a\otimes c)(b\otimes d)^tT=(a\otimes c)(d\otimes b)^t=ad^t\otimes cb^t$ and  $S(AT)T=S(ad^t\otimes cb^t)T=(F(ad^t)F(cb^t)^t)T=((a\otimes d)(c\otimes b)^t)T=(a\otimes d)(b\otimes c)^t=ab^t\otimes dc^t.$
Therefore, $A^{t_2}=ab^t\otimes dc^t=S(AT)T$.

For the other formula,   $S(A^{t_2})=S(ab^t\otimes dc^t)=(a\otimes b)(d\otimes c)^t$
and $$S(A)T=S(ab^t\otimes cd^t)T=(a\otimes b)(c\otimes d)^tT=(a\otimes b)(d\otimes c)^t.$$
\end{proof}

\begin{remark}
These formulas can be rewritten using the $*-$product $($defined in \cite{cariello}$)$ as $ A* T=S(AT)T$ and $S(A*T)=S(A)T$, because $A^{t_2}=A* T$. They show a very interesting connection between the flip operator, the partial transposition, the usual matricial product and the $*-$product.  \end{remark}

\begin{theorem}\label{theoremGuhne}$($T\'oth and G\"uhne's theorem \cite{guhne}$)$. Let $A\in M_k\otimes M_k$ be a positive semidefinite Hermitian matrix and suppose $AT=TA=A$. If $A$ is SPC then $A$ is PPT. $($Remind that $T$ is the flip operator$)$.
\end{theorem}
\begin{proof}

If $A$ is SPC then $A$ has a Hermitian Schmidt decomposition $\sum_{i=1}^n\lambda_i\gamma_i\otimes\gamma_i$, where $\lambda_i>0$.
Now, $A^{t_2}=S(AT)T$, by lemma \ref{prop2S}.
By hypothesis $A^{t_2}=S(A)T$ and, by lemma \ref{prop2S}, we get $S(A)T=S(A^{t_2})$.

Let $v_i=F(\gamma_i)$. Therefore, $A^{t_2}=S(A^{t_2})=S(\sum_{i=1}^n\lambda_i\gamma_i\otimes\gamma_i^t)=\sum_{i=1}^n\lambda_iv_i\overline{v_i}^t$.
Therefore $A^{t_2}$ is positive semidefinite and $A$ is PPT.

\end{proof}
\section{Tensor Rank 3}

Recently, in \cite{cariello}, the author proved that every positive semidefinite matrix with tensor rank 2, in $M_k\otimes M_m$, is separable. The same result is not true for matrices with tensor rank 3, we provide a counterexample in section 5.

Although, in this section, we prove that every positive semidefinite Hermitian matrix  with tensor rank 3, in $M_2\otimes M_m$, is separable $($theorem \ref{tensorrank3}$)$. This theorem is a consequence of the theorem \ref{leftinvariant} obtained in \cite{kraus}. We shall use this theorem to prove that every SPC matrix in $M_2\otimes M_2$ is PPT and therefore it is separable. Thus, symmetry implies separability in $M_2\otimes M_2$.

\begin{theorem}\label{leftinvariant}
Let $A\in M_2\otimes M_m$ be a positive semidefinite Hermitian matrix. If $A$ is invariant by the left partial transposition $((\cdot)^t\otimes Id(A)=A)$ then $A$ is separable.
\end{theorem}
\begin{proof}
See \cite{kraus}.
\end{proof}

\begin{theorem}\label{tensorrank3}
Let $A\in M_2\otimes M_m$ be a positive semidefinite Hermitian matrix. If $A$ has tensor rank smaller or equal to 3 then $A$ is separable.
\end{theorem}
\begin{proof}  If $A$ has tensor rank smaller or equal to 2 then by theorem 4.7 in \cite{cariello}, $A$ is separable. Let us suppose that $A$ has tensor rank 3.

 It is possible to find a  Hermitian decomposition, $A=\sum_{i=1}^3A_i\otimes B_i$, such that $A_1$ and $B_1$ are positive semidefinite.
Define $A(\epsilon)= A_1'\otimes B_1 +\sum_{i=2}^3A_i\otimes B_i$, such that $A_1'=A_1+\epsilon Id$, $\epsilon>0$. Notice that $A_1'$ is positive definite and let $A_1'=RR^*$ where $R$ is an invertible matrix. 
Now let us prove that $A(\epsilon)$ is separable for every $\epsilon>0$
If for some $\epsilon>0$, $A(\epsilon)$ has tensor rank 1 then $A(\epsilon)$ is separable. If $A(\epsilon)$ has tensor rank 2 then $A(\epsilon)$ is also separable by theorem 4.7 in \cite{cariello}. Let us suppose that $A(\epsilon)$ has tensor rank 3.

Let $C=(R^{-1}\otimes Id)A(\epsilon)((R^{-1})^*\otimes Id)=Id\otimes B_1+\sum_{i=2}^3A_i'\otimes B_i$, where $A_i'=R^{-1}A_i(R^{-1})^*$,  for $i=2,3$. 

Since $A_2'$ is a Hermitian matrix, there is an unitary matrix $U$ and real diagonal matrix $D$ such that $A_2'=UDU^*$. Notice that $D\neq \lambda Id$, otherwise $A_2'=\lambda Id$ and $A_2=\lambda A_1'$, which is not possible since $A(\epsilon)$ has tensor rank 3.

Let $E=(U^*\otimes Id)C(U\otimes Id)=Id\otimes B_1+D\otimes B_2+A_3''\otimes B_3$, where $A_3''=U^*A_3'U$.

Since $D\neq \lambda Id$, any diagonal matrix in $M_2$ can be written as linear combination of $D$ and $Id$. Let $D'$ be the  diagonal of $A_3''$. Notice that $D'$ is a real diagonal matrix, since $A_3''$ is Hermitian. Write $D'=aId+cD$, where $a,c$ are real numbers.

Thus, $E=Id\otimes (B_1+a B_3)+D\otimes (B_2+c B_3)+ A_3'''\otimes B_3$, where $A_3'''=A_3''-D' =\left(\begin{array}{cc}
0 & \overline{b}\\
b & 0
\end{array}\right)$.

Notice that $b\neq 0$, otherwise $E$ would have tensor rank 2 and $A(\epsilon)$ would have tensor rank 2.
 
Let $V=\left(\begin{array}{cc}
1& 0\\
0 & \overline{b}
\end{array}\right)$ and consider $F=(V\otimes Id)E(V^*\otimes Id)$. Notice that $VV^*$ and $VDV^*$ are diagonal matrices and $VA_3''' V^*= \left(\begin{array}{cc}
0 & \overline{b}b\\
\overline{b}b & 0
\end{array}\right)$ is symmetric too. Thus, $F$ is positive semidefinite Hermitian matrix in $M_2\otimes M_m$ invariant by the left partial transposition. Therefore, by theorem \ref{leftinvariant}, $F$ is separable. 

Therefore $A(\epsilon)$ is separable for every $\epsilon>0$, because $A(\epsilon)=(RUV^{-1}\otimes Id)F((V^{-1})^*U^*R^*\otimes Id)$.

Since the set of separable matrices is closed then $\displaystyle\lim_{\epsilon\rightarrow 0+}A(\epsilon)=A$ is separable. 
 
\end{proof}
\begin{remark} Notice that the maximum tensor rank in $M_2\otimes M_m$ is 4. Thus,  in order to solve the separability problem in  $M_2\otimes M_m$, we only need to deal with matrices with tensor rank 4.
\end{remark}

\section{SPC is PPT in $M_2\otimes M_2$}

In this section, we prove that every SPC matrix is PPT in $M_2\otimes M_2$. The proof of this result relies on theorem \ref{inequality}. However, in $M_3\otimes M_3$, there exists a SPC matrix which is not PPT and we shall present this counterexample in the next section. 
In the last section, we provide a non trivial example of a family of SPC matrices in $M_{k}\otimes M_k$ $(k\in \mathbb{N})$ that are also PPT.

\begin{lemma}\label{shape}
Let $A\in M_2\otimes M_2$ be a SPC matrix. If $A$ has tensor rank 4 then $A$ can be written as $A=\lambda Id\otimes Id+ D\otimes D+\gamma\otimes\gamma+\delta\otimes\delta$, where $D$ is a real diagonal matrix, $\gamma, \delta$ are Hermitian matrices and $\lambda$ a positive real number.
\end{lemma}

\begin{proof}
Since $A$ is a SPC matrix, let $A=\sum_{i=1}^4\lambda_i\gamma_i\otimes \gamma_i$ be a Hermitian Schmidt decomposition of $A$ with $\lambda_i>0$. 
Now $A^{t_2}=\sum_{i=1}^4\lambda_i\gamma_i\otimes \gamma_i^t=\sum_{i=1}^4\lambda_i\gamma_i\otimes \overline{\gamma_i}.$ Let $S$ be the linear transformation defined in \ref{defS} and let $v_i=F(\gamma_i)$.
Notice that $S(A^{t_2})=\sum_{i=1}^4\lambda_i v_i\overline{v_i}^t$ is a spectral decomposition of $S(A^{t_2})$. Since $\lambda_i>0$, $S(A^{t_2})$ is a positive definite Hermitian matrix in $M_2\otimes M_2$.

Thus, the vector $u=\sum_{i=1}^2e_i\otimes e_i\in\Im(S(A^{t_2}))$, where $\{e_1,e_2\}$ is the canonical basis of $\mathbb{C}^2$. Therefore exists a positive real number $\lambda$ such that $B=S(A^{t_2})-\lambda uu^t$ is a positive semidefinite Hermitian matrix of rank 3. 
Notice that $u\notin\Im(B)$, otherwise would exist a $\epsilon>0$ such that $B-\epsilon uu^t$ is a positive semidefinite Hermitian matrix of rank 2 and $S(A^{t_2})= \lambda uu^t+B= (\lambda +\epsilon) uu^t+ B-\epsilon uu^t$. Therefore $S(A^{t_2})$ would have rank 3.  A contradiction.
Next, since $B=S(A^{t_2})-\lambda uu^t=\sum_{i=1}^4\lambda_i v_i\overline{v_i}^t-\lambda uu^t$ then $\Im(B)$ has a basis with Hermitian eigenvectors of $B$ by lemma \ref{hermitianeigenvector}. Let $W$ be the real span of this basis. Notice that the $\dim(W)=3$.
Let $H$ be the real vector space of the Hermitian vectors of $\mathbb{C}^2\otimes\mathbb{C}^2$. Let $V$ be the real span of  $\{e_1\otimes e_1,e_2\otimes e_2\}$. Notice that $V+W\subset H$ and the $\dim(H)=4$, $\dim(V)=2$ and $\dim(W)=3$. Therefore exists a vector $d=d_1e_1\otimes e_1+ d_2e_2\otimes e_2\in V\cap W$. Remind that, by definition of $V$, $d_1,d_2\in\mathbb{R}$. Since $W\subset\Im(B)$, then $d=d_1e_1\otimes e_1+ d_2e_2\otimes e_2\in\Im(B)$, with $d_1,d_2\in\mathbb{R}$.

This vector $d$ is not a multiple of $u$, because $d\in\Im(B)$ and $u$ does not.
Again we can find $\mu>0$, such that $B-\mu dd^t$  is a positive semidefinite Hermitian matrix of rank 2 and satisfies the conditions of lemma \ref{hermitianeigenvector}.

Thus,  we can write $B-\mu dd^t=ar\overline{r}^t+bs\overline{s}^t$, where $a,b$ are positive real numbers and $r,s$ are Hermitian vectors, by lemma \ref{hermitianeigenvector}. Thus,  $S(A^{t_2})=\lambda uu^t+\mu dd^t+ar\overline{r}^t+bs\overline{s}^t$. 

By lemma \ref{prop1S}, $A^{t_2}=S^2(A^{t_2})=S(\lambda uu^t+\mu dd^t+ar\overline{r}^t+bs\overline{s}^t)=$ $$\lambda F^{-1}(u)\otimes F^{-1}(u)+\mu F^{-1}(d)\otimes F^{-1}(d)+aF^{-1}(r)\otimes F^{-1}(\overline{r})+bF^{-1}(s)\otimes F^{-1}(\overline{s})=$$ $$\lambda Id\otimes Id+(\sqrt{\mu} D)\otimes (\sqrt{\mu}D)+(\sqrt{a}F^{-1}(r))\otimes (\overline{\sqrt{a}F^{-1}(r)})+(\sqrt{b}F^{-1}(s))\otimes (\overline{\sqrt{b}F^{-1}(s)}).$$
 
Finally $A=$ $$\lambda Id\otimes Id+(\sqrt{\mu} D)\otimes (\sqrt{\mu}D)+(\sqrt{a}F^{-1}(r))\otimes (\sqrt{a}F^{-1}(r))+(\sqrt{b}F^{-1}(s))\otimes (\sqrt{b}F^{-1}(s)).$$

\end{proof}

\begin{theorem}\label{inequality} Let $A\in M_2\otimes M_2$ be a SPC matrix. If $\sum_{i} A_i\otimes A_i$ is a Hermitian decomposition of $A$ then $\sum_{i} A_i\circ A_i^t$ is a positive semidefinite Hermitian matrix in $M_2$. $($Remind that $\circ$ denotes the Schur Product$)$
\end{theorem}

\begin{proof}
Let $A$ be a SPC matrix. If the tensor rank of $A$ is smaller or equal to 3 then, by theorem \ref{tensorrank3}, $A$ is separable and therefore PPT. Thus, $\sum_{i} A_i\circ A_i^t$ is positive semidefinite, since it is a principal submatrix of the positive matrix $\sum_{i} A_i\otimes A_i^t$.

Let us suppose  that $A$ has tensor rank 4. By lemma \ref{shape}, we can write $A=\lambda Id\otimes Id+ D\otimes D+\gamma\otimes\gamma+\delta\otimes\delta$, where $D$ is a real diagonal matrix, $\gamma, \delta$ are Hermitian matrices and $\lambda$ is a positive real number.
 
 Let  $d_1,d_2$ be the real numbers in the diagonal of D and $\{e_1,e_2\}$ be the canonical basis of $\mathbb{C}^2$. Notice that $$B=\sqrt{\lambda Id+D^2} \otimes \sqrt{\lambda Id+D^2}-\lambda Id\otimes Id- D\otimes D$$ is a positive semidefinite diagonal matrix, because $$B(e_i\otimes e_j)=(\sqrt{\lambda+d_i^2}\sqrt{\lambda+d_j^2}-\lambda-d_id_j)(e_i\otimes e_j)$$
and $\sqrt{\lambda+d_i^2}\sqrt{\lambda+d_j^2}-\sqrt{\lambda}\sqrt{\lambda}-d_id_j\geq0$, by Cauchy-Schwarz inequality.

Therefore $B+A=\sqrt{\lambda Id+D^2} \otimes \sqrt{\lambda Id+D^2}+\gamma\otimes\gamma+\delta\otimes\delta$ is positive with tensor rank 3 in $M_2\otimes M_2$. By theorem \ref{tensorrank3}, this matrix is separable and PPT.

Thus, $\sqrt{\lambda Id+D^2} \otimes \sqrt{\lambda Id+D^2}+\gamma\otimes\gamma^t+\delta\otimes\delta^t$ is positive  and its principal submatrix $\sqrt{\lambda Id+D^2} \circ \sqrt{\lambda Id+D^2}+\gamma\circ\gamma^t+\delta\circ\delta^t$ is also positive, but $\sqrt{\lambda Id+D^2} \circ \sqrt{\lambda Id+D^2}=\lambda Id+D^2=\lambda Id\circ Id+D\circ D$. Therefore $\lambda Id\circ Id+D\circ D+\gamma\circ\gamma^t+\delta\circ\delta^t$ is a positive semidefinite Hermitian matrix.

But $\sum_{i} A_i\circ A_i^t=\lambda Id\circ Id+D\circ D+\gamma\circ\gamma^t+\delta\circ\delta^t$, because these matrices are the same principal submatrix of $A^{t_2}$.
\end{proof}

\begin{theorem}\label{SPC}
Every SPC matrix in $M_2\otimes M_2$ is PPT. Thus, every SPC is separable in $M_2\otimes M_2$.
\end{theorem}
\begin{proof}
Since $A$ is a SPC matrix, let $A=\sum_{i}\lambda_i\gamma_i\otimes \gamma_i$ be a Hermitian Schmidt decomposition of $A$ with $\lambda_i>0$.  

 Suppose $A^{t_2}$ has a negative eigenvalue. Since $A^{t_2}=\sum_{i}\lambda_i\gamma_i\otimes \gamma_i^t$,  we can affirm that exists a Hermitian eigenvector $v\in\mathbb{C}^2\otimes\mathbb{C}^2$ associated to this negative eigenvalue, by lemma \ref{hermitianeigenvector}. Let $\sum_{i=1}^2\lambda_i v_i\otimes \overline{v_i}$ be  a spectral decomposition of $v$ $(v_1,v_2$ are orthonormal and $\lambda_i \in\mathbb{R})$. 

Consider the unitary matrix $R\in M_2$ such that $v_1$ is the first column and  $v_2$ is the second. Thus, $v=(R\otimes\overline{R})w$, where $w=\sum_{i=1}^2\lambda_i e_i\otimes e_i$ and $\{e_1,e_2\}$ is the canonical basis of $\mathbb{C}^2$.

 Then
$0>tr(A^{t_2}v\overline{v}^t)=tr(A^{t_2}(R\otimes\overline{R})ww^t(\overline{R}^t\otimes R^t))=tr((\overline{R}^t\otimes R^t)A^{t_2}(R\otimes\overline{R})ww^t)=tr((\overline{R}^t\otimes \overline{R}^t)A(R\otimes R)(ww^t)^{t_2})$.

Since $A$ is SPC then $B=(\overline{R}^t\otimes \overline{R}^t)A(R\otimes R)$ is also SPC, by corollary \ref{defSPC2}. Let $\sum_{s} A_s\otimes A_s$ be a Hermitian decomposition of $B$.
Notice that $(ww^t)^{t_2}=\sum_{i,j=1}^2\lambda_{i}\lambda_{j} e_ie_j^t\otimes e_je_i^t$.

Finally, $0>tr(B(ww^t)^{t_2})=\sum_{i,j=1}^2\sum_s\lambda_{i}\lambda_{j} tr(A_se_ie_j^t)tr(A_se_je_i^t)=\lambda^t(\sum_s A_s\circ A_s^t)\lambda $, where $\lambda^t=(\lambda_1,\lambda_2)$. This is a contradiction with theorem \ref{inequality}. Thus, $A$ is PPT.

\end{proof}

\section{Counterexample}

In this section we show that there exists a SPC matrix in $M_3\otimes M_3$ with tensor rank 3, which is not PPT. Thus,
theorems \ref{tensorrank3} and \ref{SPC} are not true for tensor rank 3 matrices in $M_k\otimes M_m\ (k,m\geq3)$ and for SPC matrices in $M_k\otimes M_k\ (k\geq3)$, respectively. Through this section  we shall denote by $D$ and $A$ the following matrices:
\begin{center}
$D=\left(\begin{array}{ccc}
1 & 0 & 0 \\ 
0 & 3 & 0 \\ 
0 & 0 & -10
\end{array}\right) $ and $A=\left(\begin{array}{ccc}
0 & 1 & 1 \\ 
-1 & 0 & 1 \\ 
-1 & -1 & 0
\end{array}\right) $.
\end{center}

\begin{lemma}\label{lemmacounter} The smallest eigenvalue of $D\otimes D+A\otimes A$ is negative and is smaller than the smallest eigenvalue of $D\otimes D-A\otimes A$, which is also negative.
\end{lemma}
\begin{proof}
The characteristic polynomial of $D\otimes D+A\otimes A$ is \begin{center}
$p(x)=-x^9+36x^8+5420 x^7+104400 x^6-427924 x^5-14134608x^4+11251344 x^3+415 328 832 x^2-1106058240 x+671846400$
\end{center} and the characteristic polynomial of $D\otimes D-A\otimes A$ is\begin{center}
$q(x)=-x^9+36x^8+5420 x^7+104400 x^6-427924 x^5-14134608x^4+10924160 x^3+415 328 832 x^2-1106058240 x+671846400$.
\end{center}

First, notice that $p(x)-q(x)=cx^3$ with $c>0$.
Next, $D\otimes D+A\otimes A$ and $D\otimes D-A\otimes A$ are real symmetric matrices, therefore $p(x)$ and $q(x)$ have only real roots.  Notice that $0$ is not a root of $p(x)$ and neither of $q(x)$.

Since $p(x)-q(x)=cx^3$ and $0$ is not a root of $p(x)$ and neither of $q(x)$ then $p(x)$ and $q(x)$ do not have a common root.
Let us write $p(x)=(-1)(x-r_1)\ldots(x-r_9)$ and $q(x)=(-1)(x-s_1)\ldots(x-s_9)$.
If  $p(x)$ and $q(x)$ had only positive roots then all the coefficients of $p(-x)$ and $q(-x)$ would be positive.
Notice that the coefficient of $x^7$ of $p(-x)$ and $q(-x)$ are negative. Therefore $p(x)$ and $q(x)$ have negative roots. 
Let $m_p$ be the smallest root of $p(x)$ and $m_q$ be the smallest root of $q(x)$.
By contradiction, suppose that $m_q<m_p$. 

Since $m_q<m_p\leq r_i$ and  $p(m_q)=(-1)(m_q-r_1)\ldots(m_q-r_9)$ then $p(m_q)$ is positive as a product of 10 negative numbers, but $p(m_q)=p(m_q)-q(m_q)=c(m_q)^3<0.$ Absurd!

Therefore $m_p<m_q$, because $p(x)$ and $q(x)$ do not have a common root.

\end{proof}

\begin{proposition} \label{example}Let $m_q$ be the smallest eigenvalue of $D\otimes D-A\otimes A$. The matrix $C=|m_q| Id\otimes Id+D\otimes D+(iA)\otimes (iA)$ is a SPC matrix with tensor rank 3, but it is not PPT. Therefore it is not separable.
\end{proposition}
\begin{proof}
Since $m_q$ is the smallest eigenvalue of $D\otimes D-A\otimes A$ then $C=|m_q|Id\otimes Id+D\otimes D-A\otimes A$ is positive semidefinite. By corollary \ref{defSPC2}, $C=|m_q|Id\otimes Id+D\otimes D+(iA)\otimes (iA)$ is SPC. Since $\{Id,D, A\}$ is a linear independent set then $C$ has tensor rank 3.

Now, $|m_q|+m_p$ is an eigenvalue of $C^{t_2}=|m_q|Id\otimes Id+D\otimes D+A\otimes A$, where $m_p$ is the smallest eigenvalue of $D\otimes D+A\otimes A$. By lemma \ref{lemmacounter}, $m_q-m_p>0$ and $-m_q+m_p=|m_q|+m_p<0$.  Thus, $C$ is not PPT.
\end{proof}

\section{Non Trivial Example}

In this section we present  a non trivial family of matrices in $M_k\otimes M_k$, in which the SPC property and  the PPT property are equivalent $($proposition \ref{nontrivial}$)$. Inside this family, we discover a non trivial subfamily in which the SPC property is equivalent to separability$($proposition \ref{exampleseparable}$)$.

\begin{lemma}\label{lemmaexample} Let $A\in M_k\otimes M_k$ be a Hermitian matrix such that 
$A=\sum_{j=1}^n\gamma_j\otimes\gamma_j$, where $\gamma_j=i(B_j)$ and $B_j$ is a real anti-symmetric matrix for each $j$. If $\lambda$ is the smallest eigenvalue of $A$ then $\lambda$ is negative and $|\mu|\leq|\lambda|$ for any other eigenvalue $\mu$ of A.
\end{lemma}
\begin{proof} First, $A^{t_2}=\sum_{j=1}^n\gamma_j\otimes\gamma_j^t=\sum_{j=1}^n\gamma_j\otimes\overline{\gamma_j}$. Let $v_j=F(\gamma_j)$. Therefore $S(A^{t_2})=\sum_{j=1}^nv_j\overline{v_j}^t$ is a positive semidefinite Hermitian matrix.

Next, by lemma \ref{hermitianeigenvector}, $A^{t_2}$ has a spectral decomposition $\sum_{l=1}^m\lambda_lw_l\overline{w_l}^t$  such that $w_l$ is Hermitian for each $l$. Therefore $S(A^{t_2})=\sum_{l=1}^m\lambda_l\delta_i\otimes \delta_l^t $, where $\delta_l=F^{-1}(w_l)$, by lemma \ref{prop1S}. Notice that $\{\delta_1,\ldots,\delta_m\}$ is an orthonormal set of Hermitan matrices, because $F$ is an isometry $($See remark \ref{spectral}$)$.

Let us suppose that $|\lambda_1|=\ldots=|\lambda_s|>|\lambda_{s+1}|\geq\ldots\geq|\lambda_m|$, where $1\leq s\leq m$. 

Let us write, $S(A^{t_2})=\sum_{l=1}^m|\lambda_l|\delta_l\otimes (\frac{\lambda_l}{|\lambda_l|}\delta_l^t) $ . Now this is a Hermitian Schmidt decomposition of $S(A^{t_2})$.

Now, by lemma 2.9 of \cite{cariello}, $D=\sum_{l=1}^s\delta_l\otimes \frac{\lambda_l}{|\lambda_l|}\delta_l^t $ is positive semidefinite.

By contradiction, suppose that $\lambda_l<0$, for $1\leq l\leq s$, then $\frac{\lambda_l}{|\lambda_l|}=-1$ and
$tr(D)=tr(\sum_{l=1}^s\delta_l\otimes (-1)\delta_l^t )=-\sum_{l=1}^str(\delta_l)tr(\delta_l^t)<0$. This is a contradiction with the positivity of $D$.

Therefore we can suppose that $\lambda_1>0$ and notice that $|\lambda_l|\leq\lambda_1$, for $1\leq l\leq m$.

Remind that $A=\sum_{j=1}^n\gamma_j\otimes\gamma_j$, where $\gamma_j=i(B_j)$ and $B_j$ is a real anti-symmetric matrix for each $j$, thus $A^{t_2}=-A$.

Finally, since $A^{t_2}$ has the following spectral decomposition $\sum_{l=1}^m\lambda_lw_l\overline{w_l}^t$ and $A^{t_2}=-A$, then $A$ has the following spectral decomposition $A=\sum_{i=1}^m(-\lambda_l)w_l\overline{w_l}^t$.

 Notice that $|-\lambda_l|\leq|-\lambda_1|$, for $1\leq l\leq m$. Thus, the smallest eigenvalue of $A$ is $-\lambda_1$.

\end{proof}

\begin{proposition} \label{nontrivial} Let $A\in M_k\otimes M_k$ be as in lemma \ref{lemmaexample}. The matrix $C=\alpha Id\otimes Id+A$ is SPC if and only if $C$ is PPT.
\end{proposition}
\begin{proof}
Let us prove that the positivity of $C$ implies that $C$ is SPC and PPT.
By definition, SPC and PPT properties imply positivity. Thus, these three properties are equivalent for this type of $C$.

Let $\lambda$ be the smallest eigenvalue of $A$. By lemma \ref{lemmaexample}, $\lambda$ is negative. If $C$ is positive then $\alpha\geq|\lambda|$.

This matrix is SPC by corollary \ref{defSPC2}.

Now $C^{t_2}=\alpha Id\otimes Id+A^{t_2}=\alpha Id\otimes Id-A$. The eigenvalues of $A$, by lemma \ref{lemmaexample}, have absolute value smaller or equal to $|\lambda|$, therefore $C^{t_2}$ is positive semidefinite. Thus, $C$ is PPT.

\end{proof}

In the next theorem, we show that the SPC property is equivalent to separability for certain matrices of the same type described in  proposition \ref{nontrivial}. In order to provide this example, we need the following lemma.

Denote by $Sym(m)$ the subspace of the symmetric matrices in $M_{m}$ and by $ASym(m)$ the subspace of the anti-symmetric matrices in $M_{m}$.

\begin{lemma}\label{basis} Exists an orthonormal basis of $Sym(2^n)$ formed by real symmetric  matrices such that the absolute value of all their eigenvalues  is $\frac{1}{\sqrt{2^n}}$. Exists an orthonormal basis of $ASym(2^n)$ formed by real anti-symmetric  matrices such that the absolute value of all their eigenvalues  is $\frac{1}{\sqrt{2^n}}$.
\end{lemma}
\begin{proof}
The proof is by induction on $n$. 
If $n=1$, the basis of $Sym(2)$ and $ASym(2)$ required are

$\left\{S_1'=\left(\begin{array}{cc}
\frac{1}{\sqrt{2}} & 0 \\ 
0 & \frac{1}{\sqrt{2}}
\end{array}\right), 
S_2'=\left(\begin{array}{cc}
\frac{1}{\sqrt{2}} & 0 \\ 
0 & \frac{-1}{\sqrt{2}}
\end{array}\right),
S_3'=\left(\begin{array}{cc}
0 & \frac{1}{\sqrt{2}}  \\ 
\frac{1}{\sqrt{2}} & 0
\end{array} \right) \right\}$ and
$\left\{A'=\left(\begin{array}{cc}
0 & \frac{1}{\sqrt{2}}  \\ 
\frac{-1}{\sqrt{2}} & 0
\end{array} \right)\right\}$.

Suppose the result is true for $n=k-1$. 

Let $m= 2^{k-1}$ and let $\{S_1,\ldots,S_{\frac{m(m+1)}{2}}\}$ be the orthonormal basis of $Sym(m)$ announced in this theorem. Let $\{A_1,\ldots,A_{\frac{m(m-1)}{2}}\}$ be the orthonormal basis of $ASym(m)$ announced in this theorem.

Consider the following decompositions:
\begin{center}
$Sym(2^k)=[Sym(2)\otimes Sym (2^{k-1})]\oplus [ASym(2)\otimes ASym (2^{k-1})]$,
$ASym(2^k)=[Sym(2)\otimes ASym (2^{k-1})]\oplus [ASym(2)\otimes Sym (2^{k-1})]$.
\end{center}

Therefore, the set $\{S'_i\otimes S_j,\ A'\otimes A_s\ |\ 1\leq i\leq 3, 1\leq j\leq\frac{m(m+1)}{2},  1\leq s\leq\frac{m(m-1)}{2} \}$
is an orthonormal basis of $Sym(2^k)$. The eigenvalues of $S_i'\otimes S_j$ and
$A'\otimes A_s$ are the product of the eigenvalues of $S'_i, S_j$ and $A', A_s$, respectively. Therefore the absolute value of all their eigenvalues is $\frac{1}{\sqrt{2}}\times \frac{1}{\sqrt{2^{k-1}}}=\frac{1}{\sqrt{2^{k}}}$.

Next, the set $\{S_i'\otimes A_s,\ A'\otimes S_j\ |\ 1\leq i\leq 3, 1\leq j\leq\frac{m(m+1)}{2},  1\leq s\leq\frac{m(m-1)}{2} \}$
is an orthonormal basis of $ASym(2^k)$, by the decomposition above. The eigenvalues of $S_i'\otimes A_s$ and $A'\otimes S_j$ are the product of the eigenvalues of $S_i', A_s$ and $A', S_j$, respectively. Therefore the absolute value of all their eigenvalues is $\frac{1}{\sqrt{2}}\times \frac{1}{\sqrt{2^{k-1}}}=\frac{1}{\sqrt{2^{k}}}$.
\end{proof}

\begin{theorem}\label{exampleseparable} Let $\{e_1,\ldots,e_{2^n}\}$ be the canonical basis of $\mathbb{C}^{2^n}$. Let $u=\sum_{l=1}^{2^n}e_l\otimes e_l\in\mathbb{C}^{2^n}\otimes\mathbb{C}^{2^n}$. The matrix  $C=\alpha Id\otimes Id + \frac{1}{2}(T-uu^t)\in M_{2^n}\otimes M_{2^n}$ is SPC if and only if $C$ is separable. Notice that $C$ is a matrix of the same type described in proposition \ref{nontrivial}

\end{theorem}

\begin{proof}  Let us prove that the positivity of $C$ implies that $C$ is SPC and separable.
By definition, the SPC property and the separability property imply positivity. Thus, these three properties are equivalent for this type of $C$.

Let $k=2^n$.
Remind that $T$ is the flip operator whose eigenvalues are $1$ or $-1$. 
Now, $uu^t$ is a real symmetric  matrix whose eigenvalues are $k$ or $0$ and $u$ is an eigenvector of $T$ associated to 1. Therefore, $\frac{1}{2}(T-uu^t)$ is a real symmetric matrix whose eigenvalues are $-\frac{k-1}{2}$, $\frac{1}{2}$ or $-\frac{1}{2}$. Thus, we need $\alpha\geq \frac{k-1}{2}$ for the positivity of $C=\alpha Id\otimes Id + \frac{1}{2}(T-uu^t)$.

Next, if $\{S_1,\ldots S_{\frac{k(k+1)}{2}}\}$ is any orthonormal basis of $Sym(k)$, formed by real matrices, and $A_1,\ldots,A_{\frac{k(k-1)}{2}}$ is any orthonormal basis of $ASym(k)$, formed by real matrices, then 
\begin{center}
$uu^t=\sum_{l=1}^{\frac{k(k+1)}{2}}S_l\otimes S_l + \sum_{j=1}^{\frac{k(k-1)}{2}}A_j\otimes A_j$ and  $T=(uu^t)^{t_2}=\sum_{l=1}^{\frac{k(k+1)}{2}}S_l\otimes S_l - \sum_{j=1}^{\frac{k(k-1)}{2}}A_j\otimes A_j$.
\end{center}

Thus, $C=\alpha Id\otimes Id + \frac{1}{2}(T-uu^t)= \alpha Id\otimes Id - \sum_{j=1}^{\frac{k(k-1)}{2}} A_j\otimes A_j$ and  \begin{center}
$C=\alpha Id\otimes Id+\sum_{j=1}^{\frac{k(k-1)}{2}} (iA_j)\otimes (iA_j).$
\end{center}

Therefore, $C$ is SPC by lemma \ref{defSPC2}. Notice that $C$ has the  format described in the proposition \ref{nontrivial}. 

Now, we may suppose that $A_1,\ldots,A_{\frac{k(k-1)}{2}}$ is the basis constructed in lemma \ref{basis}. Therefore their eigenvalues  have absolute value equal to $\frac{1}{\sqrt{2^n}}=\frac{1}{\sqrt{k}}$. Thus, $\frac{1}{k} Id\otimes Id +(iA_j)\otimes(iA_j)$ is positive semidefinite with tensor rank 2. Therefore these matrices are separable by theorem 4.7 in \cite{cariello}. 

Next, $C=(k\alpha)(\frac{1}{k} Id\otimes Id)+\sum_{j=1}^{\frac{k(k-1)}{2}} (iA_j)\otimes (iA_j)=$\begin{center}
$(k\alpha-\frac{k(k-1)}{2}) (\frac{1}{k} Id\otimes Id)
+\sum_{j=1}^{\frac{k(k-1)}{2}} (\frac{1}{k} Id\otimes Id)+(iA_j)\otimes (iA_j).$
\end{center}

Thus, $C$ is separable as a sum of separable matrices.

\end{proof}
\begin{remark} Actually, $C=\alpha Id\otimes Id + \frac{1}{2}(T-uu^t)$, $\alpha\geq\frac{2^n-1}{2}$, is also  separable in the multipartite case. 
Notice that every matrix in the basis of $ASym(2^n)$, constructed in lemma \ref{basis}, has tensor rank $1$ in $M_2\otimes\ldots\otimes M_2\simeq M_{2^n}$ and also the $Id$. Therefore $Id\otimes Id$ and $(iA_j)\otimes (iA_j)$ have tensor rank $1$ in $M_2\otimes\ldots\otimes M_2\simeq M_{2^{2n}}$. Thus,  $\frac{1}{2^n} Id\otimes Id+(iA_j)\otimes (iA_j)$ has tensor rank smaller or equal to 2 in $M_2\otimes\ldots\otimes M_2\simeq M_{2^{2n}}$ and is positive semidefinite. By corollary 4.8 in \cite{cariello}, $\frac{1}{2^n} Id\otimes Id+(iA_j)\otimes (iA_j)$ is separable in $M_2\otimes\ldots\otimes M_2\simeq M_{2^{2n}}$. Therefore $C=\alpha Id\otimes Id + \frac{1}{2}(T-uu^t)$ is also  separable in $M_2\otimes\ldots\otimes M_2\simeq M_{2^{2n}}$.
\end{remark}

\section* {Summary}

In this paper we investigated the relationship between SPC matrices and PPT matrices. 

We proved that every SPC matrix in $M_2\otimes M_2$ is PPT and separable. Thus, in some sense symmetry implies separability in $M_2\otimes M_2$. This result follows from the fact that every density matrix with tensor rank smaller or equal to 3 in $M_2\otimes M_m$  is separable. 
However, we provided an example of SPC matrix with tensor rank 3 in $M_3\otimes M_3$ that is not PPT. 

In the last section, we showed a non trivial example of a family of matrices in $M_k\otimes M_k$
in which the SPC property is equivalent to the PPT property. Inside this family we found a subfamily, in which the SPC property is equivalent to separability.

\vspace{12pt}
\textbf{Acknowledgement.}
D. Cariello was supported by CNPq-Brazil Grant 245277/2012-9.

\end{document}